\documentclass[a4paper, 11pt]{article}

\usepackage[utf8]{inputenc}
\usepackage{hyperref}
\usepackage{amsmath}
\usepackage{amssymb}
\usepackage[usenames,dvipsnames]{color}
\usepackage{colordvi}
\usepackage{float,epsf,caption,subcaption, subfloat, multirow, paralist}
\usepackage{tikz}
\usepackage{graphicx}
\usepackage{dsfont}
\usepackage{bbold}
\usepackage{booktabs}
\usepackage{todonotes}
\usepackage{authblk}

\usetikzlibrary{decorations}
\usetikzlibrary{decorations.pathmorphing}
\usetikzlibrary{decorations.pathreplacing}
\usetikzlibrary{decorations.shapes}
\usetikzlibrary{decorations.text}
\usetikzlibrary{decorations.markings}
\usetikzlibrary{decorations.fractals}
\usetikzlibrary{decorations.footprints}

\newenvironment{proof}{\paragraph{Proof:}}{\hfill$\square$}
\newtheorem{theorem}{Theorem}[section]
\newtheorem{lemma}[theorem]{Lemma}

\newtheorem{corollary}[theorem]{Corollary}

\newtheorem{conjecture}[theorem]{Conjecture}

\newcommand{\Tr}{{\rm Tr}}

\newcommand{\F}{{\mathbb F}}

\date{}

\author{Sihem Mesnager\thanks{Department of Mathematics, University of Paris VIII,  University of Paris XIII, LAGA, UMR 7539, CNRS and Telecom ParisTech, France, email: smesnager@univ-paris8.fr}, Kwang Ho Kim\thanks{Institute of Mathematics, State Academy of Sciences,
Pyongyang, DPR Korea and PGItech Corp., Pyongyang, DPR Korea.}, Junyop Choe\thanks{Institute of Mathematics, State Academy of Sciences,
Pyongyang, DPR Korea.}, Chunming Tang\thanks{School of Mathematics and Information, China West
Normal University, Nanchong, Sichuan, 637002, China}}

\begin{document}
\title{On the Menezes-Teske-Weng's conjecture\thanks{In memory of G\'erard Cohen.}}

\maketitle

\begin{abstract}
In 2003, Alfred Menezes, Edlyn Teske and Annegret Weng presented
        a conjecture on properties of the solutions of a type of quadratic equation
        over the binary extension fields,
        which had been convinced by extensive experiments but the proof  was unknown until now. 
We prove that this conjecture is correct. Furthermore, using this proved conjecture, we have completely determined the null space of a class of linear polynomials.
\end{abstract}


\section{Introduction}

Let $p$ be a prime number and $n$ be a positive integer. The finite field with $q:=p^n$ elements is denoted by $\F_{p^n}$, which  can be viewed as an $n$-dimensional vector space over $\F_{p}$, and it is denoted by $\F_p^n$.
The trace function $\Tr: \mathbb {F}_{p^n} \rightarrow \mathbb {F}_{p}$ is defined as
\begin{displaymath}
\Tr_{p^n/p}(x) =\sum_{i=0}^{ n-1}
  x^{p^{i}}=x+x^{p}+x^{p^2}+\cdots+x^{p^{n-1}},
  \end{displaymath}
 which is called \textit{the absolute trace} of $x\in \mathbb {F}_{p^n}$,
 and also denoted by $\Tr^n_1(x)$. More general, the trace function $\Tr: \mathbb {F}_{q^n} \rightarrow \mathbb {F}_{q}$ is defined as
 \begin{displaymath}
\Tr_{q^n/q}(x) =\sum_{i=0}^{ n-1}
  x^{q^{i}}=x+x^{q}+x^{q^2}+\cdots+x^{q^{n-1}}.
  \end{displaymath}
  
Recall the transitivity property of the trace function:
  $\Tr_{k /m}\circ \Tr_{m/ n} = \Tr_{k/ n}$ provided that $k\vert m$ and $m\vert n$.\\
  
The problem of computing discrete logarithms in groups is fundamental to cryptography: it underpins the security of widespread cryptographic protocols for key exchange \cite{Diffie-Hellman1976}, public-key encryption \cite{Cramer-Shoup, ElGamal1984}, and digital signatures \cite{Johnson-Menezes-Vanstone, Kravitz, Schnorr}.

Let $E$ be an elliptic curve over a finite field $\mathbb {F}_{q}$, where $q=p^n$ and $p$ is prime. The elliptic curve discrete logarithm problem  is the following computational problem: Given points $P,Q \in E(\mathbb {F}_{q})$ to find an integer $a$, if it exists, such that $Q = aP$. This problem is the fundamental building block for elliptic curve cryptography and pairing- based cryptography, and has been a major area of research in computational number theory and cryptography for several decades.

In \cite{mtw03}, it has been considered that if for
$b\in\mathbb{F}_{2^n}^*$ there exist
$\gamma_1,\gamma_2\in\mathbb{F}_{2^n}$ such that
$b=(\gamma_1\gamma_2)^2$ then the Discrete Logarithm Problem (DLP) on the elliptic curve $E:
y^2+xy=x^3+ax^2+b$ over $\mathbb{F}_{2^n}$ $(n=6l)$ with
$\Tr_{2^n/2} (a)=0$ can be reduced to the DLP in a subgroup of the
divisor class group of an explicitly computable curve $C$ over
$\mathbb{F}_{2^l}$ with greater genus. By testing an algorithm which
decides whether such
$\gamma_1,\gamma_2$ exist for given $b\in\mathbb{F}_{2^n}^*$ and then computes them if there exist, they conceived a conjecture.\\
\begin{conjecture}\label{con}
    (Conjecture 15 of \cite{mtw03}) Let $q=2^l(l\in \mathds{N})$  and
     $\beta\in \mathbb{F}_{q^6}^*$ . Suppose that the quadratic
    equation  $u^2+(\beta ^{q^4-1}+\beta ^{q^2-1}+1)u+\beta ^{q^2-1}=0$
    has two solutions  $u_1,u_2$ in $\mathbb{F}_{q^6}$. Then  $u_1$ and  $u_2$
    satisfy  $u_i^{q^2+1}+u_i+1=0$.
\end{conjecture}
\ \ \ They have verified the conjecture with 10000 randomly chosen
$\beta\in\mathbb{F}_{q^6}^*$ respectively for
$l=5,6,7,8,9,10,19,20,21,34,35,36,37$ \cite{mtw03}. The table below
has been built by their computer experiment.
\begin{table}{}
\caption{Experimental results \cite{mtw03}}
\begin{center}
    \begin{tabular}{|c|c|c|}
        \hline
        $l$  &  Number of equations solvable in $\mathbb{F}_{q^6}$  &
        Solutions satisfying $u^{q^2+1}+u+1=0$ \\
        \hline
        5 & 5088 & 5088\\
        6 & 4985 & 4985\\
        7 & 4924 & 4924\\
        8 & 5018 & 5018\\
        9 & 4955 & 4955\\
        10& 5013 & 5013\\
        19& 5028 & 5028\\
        20& 4993 & 4993\\
        21& 4967 & 4967\\
        34& 4956 & 4956\\
        35& 5001 & 5001\\
        36& 5053 & 5053\\
        37& 5100 & 5100\\
        \hline
    \end{tabular}
\end{center}
\end{table}
\ \ \ As for the conjecture, following lemma was only what they could
prove in \cite{mtw03}.
\begin{lemma}\label{lem01}
    (Lemma 16 of \cite{mtw03}) Assume $u_1,u_2\in \mathbb{F}_{q^6}$
    are the two solutions to
    $u^2+(\beta ^{q^4-1}+\beta ^{q^2-1}+1)u+\beta ^{q^2-1}=0$ and
    $ u_1^{q^2+1}+u_1+1=0$. Then also $ u_2^{q^2+1}+u_2+1=0$.
\end{lemma}

In this paper we prove that Conjecture \ref{con} is correct. Let
 $q=2^s$ $(s\in \mathds{N})$ and $\beta\in \mathbb{F}_{q^3}^*$. In fact,
we show that a more general statement holds:  if equation
$u^2+(\beta ^{q^2-1}+\beta ^{q-1}+1)u+\beta ^{q-1}=0$ has a solution
$u\in\mathbb{F}_{q^3}$ then $u^{q+1}+u+1=0$. Furthermore we consider
 what conditions are needed in addition to $u^{q+1}+u+1=0$ in order
that the equation has a solution $u$ in $\mathbb{F}_{q^3}$.


\section{Proof of the conjecture of Menezes-Teske-Weng}

 In this section, we will prove the following result.
\begin{theorem}\label{theo01}
    Let $q=2^s$ $(s\in \mathds{N})$ and $\beta\in \mathbb{F}_{q^3}^*$.    If
    \begin{equation}
        \label{*} u^2+(\beta ^{q^2-1}+\beta ^{q-1}+1)u+\beta ^{q-1}=0
    \end{equation}
   for $u\in \mathbb{F}_{q^3}$, then
    \begin{equation}
        \label{*1} u^{q+1}+u+1=0.
    \end{equation}
    Furthermore, under the assumption $\Tr_{q^3/q} (\beta) \neq 0$,
    \eqref{*} holds for $u\in \mathbb{F}_{q^3}$
    if and only if
    \eqref{*1}  along with
    \begin{equation}
        \label{*2}\beta u+\beta ^q+\beta\in
        \mathbb{F}_q
    \end{equation}
    holds.
\end{theorem}
\textbf{Proof.} Suppose that \eqref{*} holds for $u\in \mathbb{F}_{q^3}$. Then, since $\beta\neq 0$,
we can see $u\neq0$ from \eqref{*}.\\

By exploiting the trace mapping $\Tr_{q^3/q}:\mathbb{F}_{q^3}\rightarrow
\mathbb{F}_q$ , the quadratic equation \eqref{*} can be rewritten as
\begin{equation}
\label{01} \beta u^2+\Tr_{q^3/q}(\beta)u+\beta ^q=0 .
\end{equation}

For the solutions $u_1,u_2$ to \eqref{01}, $\frac{1}{u_1}$ and
$\frac{1}{u_2}$ are two solutions to
\begin{equation}
\label{02} \beta ^q v^2+\Tr_{q^3/q}(\beta)v+\beta =0 .
\end{equation}

  Now, we will show that $(u_1+1)^q$ and $(u_2+1)^q$ are also solutions to
  \eqref{02}.\\
\ \ \
In fact, substituting  $(u_i+1)^q$ to the left side of
\eqref{02} gives
\begin{quote}
        $\beta ^q(u_i+1)^{2q}+\Tr_{q^3/q}(\beta)(u_i+1)^{q}+\beta$\\
        $=(\beta u_i^2+\beta)^q+\Tr_{q^3/q}(\beta)(u_i^q+1)+\beta$\\
        $=(\Tr_{q^3/q}(\beta)u_i+\beta^q+\beta)^q+\Tr_{q^3/q}(\beta)(u_i^q+1)+\beta$\\
        $=\Tr_{q^3/q}(\beta)^qu_i^q+\beta^{q^2}+\beta^q+\Tr_{q^3/q}(\beta)u_i^q+\Tr_{q^3/q}(\beta)+\beta$\\
        $=\Tr_{q^3/q}(\beta)u_i^q+\beta^{q^2}+\beta^q+\Tr_{q^3/q}(\beta)u_i^q+\Tr_{q^3/q}(\beta)+\beta$\\
        $=\beta^{q^2}+\beta^q+\Tr_{q^3/q}(\beta)+\beta=0,$\\
\end{quote}
where the first and third equalities was derived using properties of
finite fields and the second equality using the fact that $u_i$  is
a solution to \eqref{01} and the fourth and sixth equalities using
the
definition of the trace mapping.\\

 Hence, there are two possibilities. \\
\\
Case 1: $(u_1+1)^q=\frac{1}{u_1} , (u_2+1)^q=\frac{1}{u_2}$\\

    In this case, \eqref{*1} holds evidently.\\
\\
Case 2: $(u_1+1)^q=\frac{1}{u_2} , (u_2+1)^q=\frac{1}{u_1}$\\

    Substituting the first equality to the second equality,
    $(\frac{1}{(u_1+1)^q}+1)^q=\frac{1}{u_1}$ is obtained.
    From this equality and properties of finite fields, $(u_1+1)^{q^2+1}=u_1$
    or equivalently
    \begin{equation}
        \label{03} u_1^{q^2+1}+u_1^{q^2}+1=0
    \end{equation}
is followed. \\

Powering  $q-th$ to the both sides of \eqref{03},
 we get $u_1^{q^3+q}+u_1^{q^3}+1=0$.
Since  $u_1^{q^3}=u_1$ from $u_1\in \mathbb{F}_{q^3}$, it follows
that $u_1^{q+1}+u_1+1=0$. Similarly $u_2^{q+1}+u_2+1=0$.
Therefore \eqref{*1} holds in Case 2.

 On the
other hand, from $u_1^{q+1}+u_1+1=0$ and the first condition of Case
2, we can get $u_2=\frac{1}{u_1^q+1}=\frac{u_1}{u_1^{q+1}+u_1}=u_1$.
Hence $u_1$ and $u_2$ equal $\beta ^\frac{q-1}{2}$ by the well known
property of solutions
of quadratic equations and  $Tr_{q^3|q}(\beta)=0$ is followed from \eqref{*}.

In other words, Case 2 represents $Tr_{q^3|q}(\beta)=0$.\\

Next, we will prove \eqref{*2} under the condition $Tr_{q^3|q}(\beta)\neq 0$.\\

  For a solution $u$ to \eqref{*}, setting $u=\frac{Tr_{q^3|q}}{\beta}v$
for some $v\in F_q$, from \eqref{01} and \eqref{*1}, we get
$v^2+v=\frac{\beta ^{q+1}}{Tr_{q^3|q}(\beta)^2}$ and $v^{q+1}+\frac{\beta
^q}{Tr_{q^3|q}(\beta)}v=\frac{\beta ^{q+1}}{Tr(\beta
^2)}$ respectively.\\
 So, $v^2+v=v^{q+1}+\frac{\beta
^q}{Tr_{q^3|q}(\beta)}v$, i.e., $v^q+v+(1+\frac{\beta ^q}{Tr_{q^3|q}(\beta)})=0$
since $u,v\neq 0$. Therefore,
$(\frac{\beta}{Tr_{q^3|q}(\beta)}u)^q+(\frac{\beta}{Tr_{q^3|q}(\beta)}u)+(\frac{\beta
^{q^2}+\beta}{Tr_{q^3|q}(\beta)})=0$, i.e.
\begin{equation}
    \label{04} (\beta u)^q+(\beta u)=\beta +\beta ^{q^2}.
\end{equation}

By substitution, it is easily checked that $u_0=\beta ^{q-1}+1$ is a
solution to \eqref{04}. From the well known property of linearized
polynomial (\cite{m1993}), the set of all solutions to \eqref{04}
are $u_0+\frac{1}{\beta}\mathbb{F}_q$. In other words, if $u$ is a
solution to $u^2+(\beta ^{q^2-1}+\beta^{q-1}+1)u+\beta^{q-1}=0$ then
$\beta u+\beta
u_0=\beta u+\beta ^q+\beta \in \mathbb{F}_q$.\\

Conversely, suppose that $\beta u+\beta ^q+\beta \in \mathbb{F}_q$
and $u^{q+1}+u+1=0$.
 Then from $\beta u+\beta ^q+\beta
\in \mathbb{F}_q$, $(\beta u+\beta ^q+\beta)^q=\beta u+\beta
^q+\beta$ and so $(\beta u)^q+\beta u+(\beta ^{q^2}+\beta)=0$, i.e.,
\begin{equation}
   \label{05} u^q+\frac{u}{\beta
   ^{q-1}}+(\frac{Tr_{q^3|q}(\beta)}{\beta^q}+1)=0.
\end{equation}
Substituting $u^q=\frac{u+1}{u}$ which is obtained from
$u^{q+1}+u+1=0$ to \eqref{05}, we get $\frac{1}{u}+\frac{u}{\beta
^{q-1}}+\frac{Tr_{q^3|q}(\beta)}{\beta ^q}=0$ i.e. $u^2+(\beta
^{q^2-1}+\beta ^{q-1}+1)u+\beta ^{q-1}=0$. $\Box$

\begin{corollary}\label{cor01}
    The Menezes-Teske-Weng Conjecture is correct.
\end{corollary}
\textbf{Proof.} Letting $s=2l$ and $q\acute{}=2^l$ in the setting of
Theorem \ref{theo01} gives the corollary.\ \ \ $\Box$
\\

Theorem \ref{theo01} presents a necessary and sufficient condition
for the quadratic equation \eqref{*} to have two solutions in
$\mathbb{F}_{q^3}$. On the other hand, a classical result in the
theory of finite fields says that the quadratic equation \eqref{*}
has two solutions in $\mathbb{F}_{q^3}$ if and only if
$Tr_{q^3|2}(\frac{\beta ^{q+1}}{Tr_{q^3|q}(\beta)^2})=0$
\cite{m1993}.

In the remainder of this section, we show that really the two
conditions $u^{q+1}+u+1=0$ and $\beta u+\beta ^q+\beta\in
\mathbb{F}_q$ can be merged to one condition $Tr_{q^3|2}(\frac{\beta
^{q+1}}{Tr_{q^3|q}(\beta)^2})=0$.


In fact, assume that $u^{q+1}+u+1=0$ and $\beta u+\beta ^q+\beta\in
\mathbb{F}_q$. Let $\alpha=\beta u+\beta ^q+\beta \in
\mathbb{F}_q$.Then $u=\beta ^{q-1}+1+\frac{\alpha}{\beta}$ and
substituting it to $u^{q+1}+u+1=0$, we get $(\beta
^{q-1}+1+\frac{\alpha}{\beta})^{q+1}+\beta
^{q-1}+\frac{\alpha}{\beta}=0,$ i.e. $(\beta
^q+\beta+\alpha)^{q+1}+\beta ^{2q}+\alpha \beta^q=0$. Since
$\alpha\in \mathbb{F}_q$, we obtain $(\beta^{q^2}+\beta
^q+\alpha)(\beta ^q+\beta+\alpha)+\beta ^{2q}+\alpha \beta^q=0$ or
$\alpha^2+Tr_{q^3|q}(\beta)\alpha+Tr(\beta^{q+1})=0$. Dividing the both
sides of the obtained expression by $Tr_{q^3|q}(\beta)^2$ gives
$(\frac{\alpha}{Tr_{q^3|q}(\beta)})^2+\frac{\alpha}{Tr_{q^3|q}(\beta)}=\frac{Tr(\beta^{q+1})}{Tr_{q^3|q}(\beta)^2}$.
So, $Tr_{q|2}(\frac{Tr_{q^3|q}(\beta
^{q+1})}{Tr_{q^3|q}(\beta)^2})=Tr_{q|2}((\frac{\alpha}{Tr_{q^3|q}(\beta)})^2+\frac{\alpha}{Tr_{q^3|q}(\beta)})=Tr_{q|2}((\frac{\alpha}{Tr_{q^3|q}(\beta)})^2)+Tr_{q|2}(\frac{\alpha}{Tr_{q^3|q}(\beta)})=0$.
From the transitivity of trace function, $Tr_{q^3|2}(\frac{\beta
^{q+1}}{Tr_{q^3|q}(\beta)^2})=Tr_{q|2}(Tr_{q^3|q}(\frac{\beta
^{q+1}}{Tr_{q^3|q}(\beta)^2}))=Tr_{q|2}(\frac{Tr_{q^3|q}(\beta
^{q+1})}{Tr_{q^3|q}(\beta)^2})=0$.

Conversely, if $Tr_{q^3|2}(\frac{\beta
^{q+1}}{Tr_{q^3|q}(\beta)^2})=0$ or equivalently
$Tr_{q|2}(\frac{Tr_{q^3|q}(\beta ^{q+1})}{Tr_{q^3|q}(\beta)^2})=0$,
then the equation (on $\alpha$)
$\alpha^2+Tr_{q^3|q}(\beta)\alpha+Tr_{q^3|q}(\beta^{q+1})=0$ has two solutions
$\alpha_1, \alpha_2$ in $\mathbb{F}_q$. Then $u_i=\beta
^{q-1}+1+\frac{\alpha_i}{\beta}(i=1,2)$ are the solutions to
$u^2+(\beta ^{q^2-1}+\beta ^{q-1}+1)u+\beta ^{q-1}=0$ and from
theorem \ref{theo01} we get $u^{q+1}+u+1=0$ and $\beta u+\beta
^q+\beta\in
\mathbb{F}_q$.\\

\section{An application of the Menezes-Teske-Weng's Conjecture to linear polynomials}

In this section, we shall study a class of linear polynomials, which are related to the equation $ u^2+(\beta ^{q^2-1}+\beta ^{q-1}+1)u+\beta ^{q-1}=0$ discussed in the previous section.

For $u\in \mathbb F_{q^3}$, we define a linear polynomial with variable $\beta$ as
\begin{align}\label{eq:L_u}
L_{u}(\beta):=u \beta^{q^2}+ (u+1)\beta^q+ (u^2+u)\beta.
\end{align}

To study the above linear polynomials, we begin with some lemmas on finite fields.
\begin{lemma}\label{lem:q+1}
Let $u$ be an element in the algebraic closure of $\mathbb F_q$ with  $u^{q+1}+u+1=0$. Then,

(i) $u^{q^2+q+1}=1$ and $u\in \mathbb F_{q^3}$;

(ii) $u^{q^2}=\frac{u^q+1}{u^q}=\frac{1}{u+1}$;

(iii) there exists $\beta\in \mathbb F_{q^3}$ such that $\beta^{q-1}=u^2$.
\end{lemma}
\begin{proof}
(i) Let $u^{q+1}+u+1=0$. Then,
\begin{align*}
u^{q^2+q+1}=& u u^{(q+1)q}\\
=& u (u+1)^q\\
=& u^{q+1}+u\\
=& 1.
\end{align*}

Thus, $u^{q^3-1}=u^{(q-1)(q^2+q+1)}=1$ and $u^{q^3}=u$. One gets that $u\in \mathbb F_{q^3}$.

(ii) One has
\begin{align*}
u^{q^2}=& \frac{u^{q^2+q}}{u^q}\\
=& \frac{u^{(q+1)q}}{u^q}.
\end{align*}
From $u^{q+1}+u+1=0$, $u^{q^2}=\frac{u^q+1}{u^q}$. Thus, $u^{q^2}=\frac{u^{q+1}+u}{u^{q+1}}=\frac{1}{u+1}$.

(iii) Since $u^{q^2+q+1}=1$, there exists $\alpha \in \mathbb F_{q^3}$ such that $u=\alpha^{q-1}$. Thus, one can choose $\beta= \alpha^{2}$.

\end{proof}

\begin{lemma}\label{lem:2eq}
Let $q=2^s$.

(i) If $s \equiv 0 \text{ or } 2 \pmod 3 $, then  there is no $u\in \mathbb F_{q^3}$ such that
$
\begin{cases}
u^{q+1}+u+1= 0
\cr u^3+u+1= 0
\end{cases}
$;

(ii) If $s \equiv 1 \pmod 3 $ and $u^3+u+1= 0$, then  $u\in \mathbb F_{q^3}$ and $u^{q+1}+u+1= 0$. Moreover, $u^2$ and $u+1$ are linearly  independent  over $\mathbb F_q$.
\end{lemma}

\begin{proof}
(i) Assume  that there exists $u\in \mathbb F_{q^3}$ such that
$
\begin{cases}
u^{q+1}+u+1= 0
\cr u^3+u+1= 0
\end{cases}
$.
One has $u^{q+1}=u^3$ and $u^{q}=u^2$. From $u^3+u+1=0$, $u\in \mathbb F_{2^3}$ and $u^{2^3}=u$.

 If $s=3k$, then $u^2=u^q=u^{2^{3k}}=u$. Thus, $u=1$, which  contradicts with $u^3+u+1= 0$.

 If $s=3k+2$, then $u^2=u^q=u^{2^{3k+2}}=u^4$. Thus, $u=1$, which contradicts with $u^3+u+1= 0$.

 (ii) Let  $s \equiv 1 \pmod 3 $ and $u^3+u+1= 0$. Thus, $u\in \mathbb F_{2^3}\subseteq \mathbb F_{q^3}$ and $u^{2^3}=u$. Write  $s=3k+1$. Then,
 $u^{q+1}=uu^{2^{3k+1}}=u^3=u+1$ and $u^{q+1}+u+1=0$.

 Assume that $u^2$ and $u+1$ are linearly  dependent  over $\mathbb F_q$. One obtains that
 $\frac{u+1}{u^2}\in \mathbb F_q$. By $u^3+u+1= 0$, $u\in \mathbb F_q \cap \mathbb F_{2^3} $ and $u^q=u^{2^3}=u$.
  Let $s=3k+1$. One has $u=u^q=u^{2^{3k+1}}=u^2$ and $u=1$, which contradicts with $u^3+u+1=0$.

\end{proof}

\begin{lemma}\label{lem:2beta}
Let $q=2^s$ and $u\in \mathbb F_{q^3}$  with  $u^{q+1}+u+1=0$. Let $\beta_1=u^q$ and $\beta_2\in \mathbb F_{q^3}$ with $\beta_2^{q-1}=u^2$. Then,
  $\beta_1$ and $\beta_2$ are linearly  independent over $\mathbb F_q$, if and only if, $s \equiv 0 \text{ or } 2 \pmod 3 $, or,  $s \equiv 1 \pmod 3 $ and
$ u^3+u+1\neq 0$.

\end{lemma}

\begin{proof}
From Part (iii) of Lemma \ref{lem:q+1}, there exists $\beta_2\in \mathbb F_{q^3}$ with $\beta_2^{q-1}=u^2$.
The statement of this lemma is equivalent to the following claim:

 $\beta_1$ and $\beta_2$ are linearly dependent over $\mathbb F_q$, if and only if, $s \equiv 1 \pmod 3 $ and
$ u^3+u+1= 0$.

Suppose that $\beta_1$ and $\beta_2$ are linearly  dependent over $\mathbb F_q$. Then, there exists $a\in \mathbb F_q^*$ such that
$\beta_2=a \beta_1$. Thus,
\begin{align*}
u^2=& \beta_2^{q-1}\\
=& (a u^q)^{q-1}\\
=& u^{q^2-q}.
\end{align*}
By Part (ii) of Lemma \ref{lem:q+1}, one obtains $u^2=\frac{u^q+1}{u^{2q}}$ and $u^{2q+2}=u^q+1$.
Thus, $u^q+1=u^2+1$ from $u^{q+1}+u+1=0$. One gets $u^q=u^2$ and $u^3+u+1=0$.
Note that $u\in \mathbb F_{q^3} \cap \mathbb F_{2^3}$. So $u^{2^3}=1$. Write $s=3k+r$ with $r\in \{0,1,2\}$.
Then, $u^q=u^{2^{3k} \cdot 2^{r}}= u^{2^r}$. Assume that $r=0 \text{ or } 2$. One gets $u=1$,
which contradicts with $u^3+u+1=0$. Hense,
$s=3k+1$.

Conversely, suppose that $s \equiv 1 \pmod 3 $ and
$ u^3+u+1= 0$. By Lemma \ref{lem:2eq}, $u^{q+1}+u+1=0$. Thus, $u^q=u^2$ and $\beta_1^{q-1}=u^2$. So $\left (\frac{\beta_2}{\beta_1} \right )^{q-1}=1$.
Thus, $\beta_2=a \beta_1$ with $a \in \mathbb F_q^*$. It completes the proof.
\end{proof}

\begin{theorem}\label{thm:L_u}
Let $q=2^s$, $u\in \mathbb F_{q^3}$ and $L_u(\beta)$ be the linear polynomial defined as in Equation (\ref{eq:L_u}). More precisely, we have

(i) if $u^{q+1}+u+1\neq 0$, then $L_u(\beta)$ is a linear permutation polynomial;

(ii) if $u^{q+1}+u+1= 0$, then  $\mathbf{dim}_{\mathbb F_q}(\mathbf{Ker} (L_u))=2$;

(iii) if $s \equiv 0 \text{ or } 2 \pmod 3 $ and  $u^{q+1}+u+1= 0$, or,  $s \equiv 1 \pmod 3 $ and
$\begin{cases}
u^{q+1}+u+1= 0
\cr u^3+u+1\neq 0
\end{cases}$, then  $\mathbf{Ker} (L_u)=\{x\beta_1+ y \beta_2: x,y \in \mathbb F_q\}$,
where $\beta_1, \beta_2 \in \mathbb F_{q^3}$ with $\beta_1=u^q$ and $\beta_2^{q-1}=u^2$;

(iv) if $s \equiv 1 \pmod 3 $ and $u^3+u+1= 0$, then,
$\mathbf{Ker} (L_u)=\{x u^2+ y (u+1): x,y \in \mathbb F_q\}$.
\end{theorem}
\begin{proof}
We first prove (iii) and (iv). Then we prove (i) and (ii).

(iii) From Part (iii) of Lemma \ref{lem:q+1}, $\beta_2$ exists. Substitute $\beta=\beta_1$ into $L_u$.
\begin{align*}
L_{u}(\beta_1)=& u \left (u^q \right )^{q^2}+ (u+1)\left (u^q \right )^q+ (u^2+u) u^q\\
=& u^2+1+(u+1)^2\\
=& 0,
\end{align*}
where the second equation follows from Part (ii) of Lemma \ref{lem:q+1}.

Substitute $\beta=\beta_2$ into $L_u$.
\begin{align*}
L_{u}(\beta_2)=& u \beta_2^{q^2}+ (u+1)\beta_2^q+ (u^2+u) \beta_2\\
=& \beta_2\left [ u \beta_2^{q^2-1}+ (u+1)\beta_2^{q-1}+ (u^2+u)  \right ]\\
=& \beta_2\left [ u \left ( u^2 \right )^{q+1}+ (u+1)u^2+ (u^2+u)  \right ]\\
=& \beta_2\left [ u (u+1)^2+ (u+1)u^2+ (u^2+u)  \right ]\\
=& 0.
\end{align*}

From Lemma \ref{lem:2beta}, $\beta_1$ and $\beta_2$ are linearly independent over $\mathbb F_q$.
Thus, $\{x\beta_1+ y \beta_2: x,y \in \mathbb F_q\} \subseteq \mathbf{Ker} (L_u)$ and $\# \mathbf{Ker} (L_u) \ge q^2$.
By $deg(L_u(\beta))=q^2$, $\# \mathbf{Ker} (L_u) \le q^2$. Then, $\mathbf{Ker} (L_u) = \{x\beta_1+ y \beta_2: x,y \in \mathbb F_q\}  $.

(iv) By Part (ii) of Lemma \ref{lem:2eq}, $u^{q+1}+u+1=0$. Thus, $u^q=u^2$. Substitute $\beta=u^2$ into $L_u$.
\begin{align*}
L_{u}(u^2)=& u \left (u^2 \right )^{q^2}+ (u+1)\left (u^2 \right )^q+ (u^2+u) u^2\\
=& u^2+(u+1)u^4+u^4+u^3\\
=& u^2\left ( u^3+u +1 \right )\\
=& 0.
\end{align*}

Substitute $\beta=u+1$ into $L_u$.
\begin{align*}
L_{u}(u+1)=& u \left (u+1 \right )^{q^2}+ (u+1)\left (u+1 \right )^q+ (u^2+u) (u+1)\\
=& u \left (u+1 \right )^{4}+ (u+1)\left (u+1 \right )^2+ (u^2+u) (u+1)\\
=& (u+1)^2 \left ( u^3+u+ u+1+u \right )\\
=& 0.
\end{align*}
From Part (ii) of Lemma \ref{lem:2eq}, $u^2$ and $u+1$ are linearly independent over $\mathbb F_q$.
Thus, $\{x u^2+ y (u+1): x,y \in \mathbb F_q\} \subseteq \mathbf{Ker} (L_u)$ and $\# \mathbf{Ker} (L_u) \ge q^2$.
By $deg(L_u(\beta))=q^2$, $\# \mathbf{Ker} (L_u) \le q^2$. Then, $\mathbf{Ker} (L_u) = \{x u^2+ y (u+1): x,y \in \mathbb F_q\}  $.

(i) Let $u^{q+1}+u+1\neq 0$. Assume that $L_u$ is not a permutation polynomial. Then, there exists $\beta \in \mathbb F_{q^3}^*$ such that $L_u(\beta)=0$, that is
\begin{align*}
u^2+(\beta ^{q^2-1}+\beta ^{q-1}+1)u+\beta ^{q-1}=0.
\end{align*}
By Theorem \ref{theo01}, $u^{q+1}+u+1= 0$, which  contradicts with $u^{q+1}+u+1\neq 0$. Hence, $L_u$ is a  permutation polynomial. (ii) It follows from (i), (iii) and (iv).

\end{proof}

\begin{corollary}
Let  $u\in \mathbb F_{q^3}$ and $L_u(\beta)$ be the linear polynomial defined as in Equation (\ref{eq:L_u}). Then,
$L_u(\beta)$ is a linear permutation polynomial, if and only if, $u^{q+1}+u+1\neq  0$.
\end{corollary}
\begin{proof}
It follows from Theorem \ref{thm:L_u}.
\end{proof}

\begin{corollary}\label{cor:q^2}
    Let $q=2^s$, $\beta\in \mathbb{F}_{q^3}^*$. Then, $u^2+(\beta ^{q^2-1}+\beta ^{q-1}+1)u+\beta ^{q-1}=0$ holds for $u\in \mathbb{F}_{q^3}$ with $u^{q+1}+u+1=0$,
    if and only if, $\beta=x \beta_1 + y \beta_2$, where $x, y\in \mathbb F_q$, and , $\beta_1=u^q$, $\beta_2^{q-1}=u^2$ if $u^3+u+1 \neq 0$, or,
     $\beta_1=u^2$, $\beta_2=u+1$ if $u^3+u+1 = 0$.

\end{corollary}
\begin{proof}
It follows from Theorem \ref{thm:L_u}.
\end{proof}

\begin{corollary}
    Let $q=2^s$. Then, for any $u\in \mathbb{F}_{q^3}$ with $u^{q+1}+u+1=0$, there exactly exist $(q^2-1)$ $\beta$'s with $\beta \in \mathbb F_{q^3}^*$
    such that $u^2+(\beta ^{q^2-1}+\beta ^{q-1}+1)u+\beta ^{q-1}=0$ holds.
\end{corollary}
\begin{proof}
It follows from Corollary \ref{cor:q^2}.
\end{proof}

\section{Conclusion}
In this article, we prove a conjecture suggested 15 years ago  by Alfred Menezes, Edlyn Teske and Annegret Weng. Such a conjecture is related to the Discrete Logarithm Problem on  elliptic curves.
We also use the proved conjecture to completely determined the null space of a class of linear polynomials.

\section* {Acknowledgements}

The authors deeply thank Alfred Menezes for checking our proof of the MTW conjecture.



\end{document}